\documentclass[12pt,a4paper,twoside]{amsart} 
\usepackage{latexsym}
\usepackage{amsmath}
\usepackage{amssymb}
\usepackage{amsfonts}
\usepackage{ifthen}
\usepackage{mathrsfs}               
\usepackage{bbm}                    
\usepackage{bbold}                  
\usepackage{textcomp}               
\usepackage{amstext}
\usepackage{enumerate}
\usepackage{amstext}
\newcommand{\C}{\mathbb{C}}
\newcommand{\N}{\mathbb{N}}

\newcommand{\R}{\mathbb{R}}

\newcommand{\Z}{\mathbb{Z}}
\newcommand{\sgn}{\mathrm{sgn}}

\newcommand{\id}{\mathbbm{1}}
\newcommand{\klg}{\leqslant} 
\newcommand{\grg}{\geqslant}          
\newcommand{\ve}{\varepsilon}
\newcommand{\vp}{\varphi}

\newcommand{\wt}[1]{\widetilde{#1}}
    
\newcommand{\SPn}[2]{\langle \,#1\,|\,#2\, \rangle} 
\newcommand{\SPb}[2]{\big\langle \,#1\,\big|\,#2\, \big\rangle} 

\newcommand{\ol}[1]{\overline{#1}} 
\newcommand{\bigO}{\mathcal{O}}    
\newcommand{\bp}{{\bf p}}
\newcommand{\V}[1]{\mathbf{#1}}
\newcommand{\valpha}{\mbox{\boldmath$\alpha$}}

\newcommand{\veps}{\mbox{\boldmath$\varepsilon$}}
\newcommand{\HP}{\mathfrak{h}}
\newcommand{\HR}{\mathcal{H}}

\newcommand{\Fock}{\mathcal{F}}
\newcommand{\core}{\mathscr{D}_0}
\newcommand{\dom}{\mathcal{D}}
\newcommand{\form}{\mathcal{Q}}
\newcommand{\spec}{\mathrm{\sigma}}

\newcommand{\DA}{D_{\mathbf{A}}}                
\newcommand{\DAT}{\widetilde{D}_\mathbf{A}}

\newcommand{\Hf}{H_f}                           
\newcommand{\HT}{\wt{H}_f}                      

\newcommand{\ad}{a^\dagger}                     

\newcommand{\TA}{T_{\bf A}}
\newcommand{\da}{d_{\bf A}}

\newcommand{\pa}{{\bf p}_{\bf A}}
\newcommand{\TAA}{{T}_{\bf A}}
\newcommand{\hf}{H_f}
\newcommand{\RA}{R_{c}(z)}
\newcommand{\RAn}{R_{\infty}(z)}
\newcommand{\ra}[1]{R^{V}_{c}(#1)}
\newcommand{\raf}{R_{c}(\rho)}
\newcommand{\ran}[1]{R^V_{\infty}(#1)}

\newcommand{\hilbert}{\mathcal{H}}

\newcommand{\al}{\mbox{\boldmath$\alpha$}}
\newcommand{\bx}{\mathbf{x}}
\newcommand{\Hn}{H_\infty}
\newcommand{\Hr}[1]{H^{#1}_c}
\newcommand{\Hnr}[1]{H^{#1}_\infty}
\newcommand{\thh}{\hat{H}_\infty}
\newcommand{\mc}{mc^2}
\newcommand{\sC}{\mathscr{C}}

\renewcommand{\Im}{\mathrm{Im}\,}
\renewcommand{\Re}{\mathrm{Re}\,}
\renewcommand{\le}{\,\leqslant\,}
\renewcommand{\ge}{\,\geqslant\,}
\newtheorem{theorem}{Theorem}[section]
\newtheorem{lemma}[theorem]{Lemma}

\newtheorem{corollary}[theorem]{Corollary}
\newtheorem{hypothesis}{Hypothesis}
\theoremstyle{remark}
\newtheorem{remark}[theorem]{Remark}

\numberwithin{equation}{section}
\title[On the non-relativistic limit of a model in quantum electrodynamics]{On the non-relativistic limit of a model in quantum electrodynamics}
\author{Edgardo Stockmeyer}
\address{Edgardo Stockmeyer\\
 Mathematisches Institut\\
Ludwig-Maximilians-Universit\"at\\Theresienstra{\ss}e 39\\
D-80333 M\"unchen, Germany.}
\email{stock@math.lmu.de}

\date{07.05.09}
\begin{document}

\begin{abstract}
We consider a (semi-)relativistic spin-$1/2$ particle interacting with quantized radiation.  The Hamiltonian has the form
$\hat{H}_c^V:=\{c^2[(\bp+\V{A})^2+{\bf \sigma}\cdot\V{B}]+(\mc)^2\}^{1/2}-\mc+V+\Hf$. Assuming that
the potential $|V|$ is  bounded with respect to the momentum $|\bp|$, we show that $\hat{H}_c^V$ converges in norm-resolvent sense to the usual Pauli-Fierz operator when $c$, the speed of light, tends to $\infty$.
\end{abstract}

\maketitle


\section{Introduction}
We consider a (semi-)relativistic electron interacting with quantized radiation in the presence of an external potential $V$. The corresponding Hamiltonian  
is given formally by 
\begin{equation}
  \label{eq:41}
  \hat{H}_c^V:=\sqrt{c^2[(\bp+\V{A})^2+{\bf \sigma}\cdot\V{B}]+(\mc)^2}-\mc+V+\Hf\,,
\end{equation}
acting on $L^2(\R^3,\C^2)\otimes\mathcal{F}$, where $\mathcal{F}$ is the bosonic Fock space.
Here $\bp$ is the momentum of the electron, $\V{A}$ is the magnetic vector potential (in the Coulomb gauge), $\V{B}$
is the magnetic field, $\Hf$ is the free field-energy operator, 
$m>0$ is the mass of the electron, $c>0$ is the speed of light, and $\V{\sigma}=(\sigma_1,\sigma_2,\sigma_3)$, where $\sigma_j$ are the usual Pauli-matrices. 

Without an external potential $V$, the Hamiltonian \eqref{eq:41} admits a fiber decomposition $\hat{H}_c^0({\bf P})$, since the total momentum ${\bf P}\in\R^3$ is conserved.  In 
\cite{MiyaoSpohn2008} Miyao and Spohn investigated the ground state of $\hat{H}_c^0({\bf P})$ (the polaron). In \cite{MatteStockmeyer2009a}, Matte and the present author studied exponential localization of the low-lying spectral subspaces for Hydrogen-like atoms modelled by  $\hat{H}_c^V$. In this paper we are interested in the 
non-relativistic limit of $\hat{H}_c^V$.

In absence of quantized radiation the non-relativistic limit is well understood (see \cite{Thaller1992}). However, in the presence of a field, only few results have been presented. Let $\mathcal{D}_c^V:=\DA-\mc+V+\Hf$ be  the Dirac operator coupled to the quantized field acting on $L^2(\R^3,\C^4)\otimes\mathcal{F}$, where $\DA$ is the usual Dirac operator minimally coupled to the vector potential $\V{A}$. Arai in \cite{Arai2003} proved that $\mathcal{D}_c^V$  converges, as $c$ tends to $\infty$, in 
{\it strong} resolvent sense to the Pauli-Fierz Hamiltonian, denoted by $\hat{H}_\infty^V$ acting on $L^2(\R^3,\C^2)\otimes\mathcal{F}$. More precisely,
\begin{equation*}
  \label{eq:42}
 {\rm s\,-}\lim_{c\to\infty}(\mathcal{D}_c^V -i)^{-1}= 
\left(\begin{array}{cc}
(\hat{H}_\infty^V-i)^{-1} & 0_{2\times 2}\\
0_{2\times 2} & 0_{2\times 2}
\end{array}\right)\,.
\end{equation*} 
 The method used cutoffs for $V$, $\Hf$, and 
$\V{A}$ in order to include them as  perturbations. (We remark that it is unlikely  to improve this convergence to norm resolvent sense, since the 
spectrum of $\mathcal{D}_c^0$ for $\V{A}=0$ is equal to the whole real line, i.e., $\sigma(D_0+\Hf)=\R$.) By the same method the non-relativistic limit of the Dirac polaron was investigated by Arai in \cite{Arai2006}.

In this paper we do not include the field energy $\Hf$ as a perturbation 
but we make strong use of the fact that $\hat{H}_c^0$ is bounded from below.
Assuming that the potential $V$ is (operator) bounded with respect to the (electric) momentum $|\bp|$, we show that, as $c$ (the speed of light) tends to infinity,   $\hat{H}_c^V$ converges in norm resolvent sense to the usual (non-relativistic) Pauli-Fierz operator $\hat{H}_\infty^V$.   

In the next section  we introduce  the models to be considered  and state precisely our main results, Theorems \ref{sa} and \ref{main} below.
\section{Definition of the models and main results}\label{def}
\noindent
Let $\HP=L^2(\R^3\times\{1,2\})$ be the one-photon Hilbert space. We use the convention
$$
k=(\V{k},\lambda)\in\R^3\times\{1,2\}\,;\qquad \int dk :=
\sum_{\lambda\in\Z_2}\int_{\R^3}d^3\V{k}\,.
$$
The space of the quantized photon field is the bosonic Fock
space
$$
\Fock\equiv\Fock[\HP]:= \bigoplus_{n=0}^\infty\HP^{\otimes_{\mathrm{s}}\,n}\ni\psi=(\psi^{(0)},\psi^{(1)},\psi^{(2)},\ldots) \,,
$$
where $\HP^{\otimes_{\mathrm{s}}\,n}$ is the $n$-fold symmetric tensor product of $\HP$ and $\HP^{\otimes_{\mathrm{s}}\,0}=\C$.  As usual we denote the
vacuum vector by $ \Omega:=(1,0,0,\dots)\in\Fock[\HP] $.  Many
calculations will be performed on the following dense subspace of
$\Fock[\HP]$,
$$
\sC_0 := \C\oplus\bigoplus_{n\in\N} C_0((\R^3\times\{1,2\})^n)\cap
\HP^{\otimes_{\mathrm{s}}\,n}\,.  
$$
The free field energy of the photons is the self-adjoint operator
given by
\begin{align*}
  \dom(\Hf) := \Big\{ &(\psi^{(n)})_{n=0}^\infty\in\Fock[\HP]\::
  \\
  & \sum_{n=1}^\infty\int\Big|\sum_{j=1}^n\omega(k_j)
  \psi^{(n)}(k_1,\dots,k_n) \Big|^2 dk_1\dots dk_n < \infty \Big\}\,,
\end{align*}
and, for $\psi\in\dom(\Hf)$, 
$$
(H_f \psi)^{(0)} = 0,\quad (H_f\psi)^{(n)}(k_1,\dots,k_n) =
\sum_{j=1}^n\omega(k_j) \psi^{(n)}(k_1,\dots,k_n)\,,\quad n\in\N\,.
$$
Here the dispersion relation $\omega$ is an almost everywhere nonzero multiplication 
operator on $\HP$ that depends only on $\V{k}$ and not on $\lambda\in\{1,2\}$. 

The annihilation operator $a(f)$ of a photon state $f\in\HP$
is, for any $\vp\in\sC_0$,  given by 
\begin{equation}\label{*}
a(f) \vp=\int \ol{f}(k) a(k) \vp dk\,,
\end{equation}
where $a(k)$ annihilates a photon with wave vector/polarization $k$,
$$
(a(k) \psi)^{(n)}(k_1,\dots,k_n) = (n+1)^{1/2}
\psi^{(n+1)}(k,k_1,\dots,k_n)\,,\quad n\in\N_0\,,
$$
almost everywhere, and $a(k) \Omega=0$. For $f\in\HP$ the creation operator 
$\ad(f)$ satisfies $\SPn{a(f) \phi}{\vp}=\SPn{\phi}{\ad(f) \vp}$, for any $\phi,\vp\in\sC_0$. We define $\ad(f)$ and $a(f)$ on their 
maximal domains. 
The following canonical commutation relations hold true
on $\sC_0$, for  any $f,g\in\HP$,
\begin{equation}\label{ee}
[a(f) , a(g)] = [\ad(f) , \ad(g)] = 0\,,\qquad
[a(f) , \ad(g)] = \SPn{f}{g} \id\,.
\end{equation}
The full Hilbert space containing electron and photon degrees of
freedom is
$$
\HR := L^2(\R^3,\C^4)\otimes\Fock\,.
$$
It contains the dense subspace,
$$
\core := C_0^\infty(\R^3,\C^4)\otimes\sC_0\,.
\quad\textrm{(Algebraic tensor product.)}
$$
We consider general form factors fulfilling the following condition:
\begin{hypothesis}\label{hyp-G}
For every $k\in(\R^3\setminus\{0\})\times\Z_2$ and  $j\in\{1,2,3\}$,
  $G^{(j)}(k)$ is a bounded continuously differentiable function,
   $\R^3\ni\V{x}\mapsto
  G^{(j)}_{\V{x}}(k)$, satisfying
\begin{equation}
  \label{def-d3}
  2\int\omega(k)^{\ell} \|\V{G}(k)\|^2_\infty dk
  \klg d_\ell^2\,,
  \qquad \ell\in\{-1,0,1,2\}\,,
\end{equation}
and
\begin{equation}\label{hyp-rotG}
  2\int\omega(k)^{-1} \|\nabla_{\V{x}}\times\V{G}(k)\|^2_\infty dk
  \klg d_{1}^2\,,
\end{equation}
for some $d_{-1},\ldots,d_2\in(0,\infty)$.  Here $\V{G}_{\V{x}}(k)=
\big(G^{(1)}_{\V{x}}(k),G^{(2)}_{\V{x}}(k),G^{(3)}_{\V{x}}(k)\big)$
and $\|\V{G}(k)\|_\infty:=\sup_{\V{x}}|\V{G}_{\V{x}}(k)|$.
\end{hypothesis}
\begin{remark}
In the specific physical situation the dispersion relation is  
$\omega^{\small{\rm ph}}(k)=\hbar c|\V{k}|$, where $\hbar>0$ is Plank's constant divided by $2\pi$. The form factor  is given by 
\begin{equation}\label{Gphys}
G^{\small{\rm ph}}_\V{x}(k) 
:= \frac{\hbar^{1/2}}{2\pi(c |\V{k}|)^{1/2}}\id_{\{|\V{k}|\klg\Lambda\}}
 e^{-i\V{k}\cdot\V{x}} \veps(k),
\end{equation}
for $(\V{x},k)\in\R^3\times(\R^3\times\{1,2\})$
with $\V{k}\not=0$.
Here 
$\Lambda>0$ is an ultraviolet cut-off
parameter.
The polarization vectors, $\veps(\V{k},\lambda)$, $\lambda\in\Z_2$,
are homogeneous of degree zero in $\V{k}$ such that
$\{\hat{\V{k}},\veps(\hat{\V{k}},0),\veps(\hat{\V{k}},1)\}$
is an orthonormal basis of $\R^3$,
for every $\hat{\V{k}}\in \mathbb{S}^2$.
This corresponds to the Coulomb gauge. 

We further remark that the  dependence on $c$ of \eqref{Gphys} and the definition of $\omega^{\small{\rm ph}}$ is not relevant to us, since we want to study the electron and 
not the photon non-relativistic limit.
\end{remark}
We introduce the self-adjoint Dirac matrices
$\alpha_1,\alpha_2,\alpha_3$, and $\beta$ 
that act on the four spinor components of an element from $\HR$.
They are given by
$$
\alpha_j := 
\begin{pmatrix}
0&\sigma_j\\\sigma_j&0
\end{pmatrix} ,\quad j\in\{1,2,3\}\,,\qquad
\beta := \alpha_0 := 
\begin{pmatrix}
\id&0\\0&-\id
\end{pmatrix}\,,
$$
where $\sigma_1,\sigma_2,\sigma_3$ denote the standard Pauli matrices,
and fulfill the relations
\begin{eqnarray}\label{Clifford1}
&&\alpha_i \alpha_j + \alpha_j \alpha_i = 2 \delta_{ij} \id\,,\qquad i,j\in\{0,1,2,3\}\,,
\\\label{Clifford2}
&& \alpha_i \alpha_j=\delta_{ij}+i\epsilon_{ijm} \Sigma_m\,,\qquad i,j\in\{1,2,3\}\,, 
\end{eqnarray}
where $\delta_{ij}$ is the Kroneker delta, $\epsilon_{ijm}$ is the Levi-Civita
antisymmetric symbol, and
\begin{equation}
  \label{eq:21}
  \Sigma_j:=
\left(\begin{array}{cc}
\sigma_j & 0\\
0 & \sigma_j 
\end{array}\right)\,.
\end{equation}

The interaction between the electron
and the photons  is  given by
\begin{equation}\label{dd}
\valpha\cdot\V{A}
 \equiv \valpha\cdot\V{A}(\bx)
 :=  a(\valpha\cdot\V{G}_\bx)+\ad(\valpha\cdot\V{G}_\bx)\,.
\end{equation}
In order to define the semi-relativistic
Pauli-Fierz operator we recall that the
free Dirac operator, minimally coupled to $\V{A}$, is given as
\begin{equation}\label{def-DA}
\DA := c \valpha\cdot(-i\hbar\nabla+ \V{A})+\beta mc^2
\,,
\end{equation}
where $m>0$ is the mass of the electron. (Note that usually 
there is a factor in front of $\V{A}$; in physical units $e/c$.
We have absorbed this factor in the definition of  $\V{A}$.)
An application of Nelson's commutator theorem shows that
$\DA$ is essentially self-adjoint
on $\core$ 
\cite{Arai2000,LiebLoss2002}.
We denote its closure starting from $\core$
again by the same symbol. 
Henceforth, we set
\begin{equation}
  \label{eq:19}
  \pa:=-i\hbar\nabla+ \V{A}\,,\qquad\mbox{and}\qquad\bp:=-i\hbar\nabla\,.
\end{equation}
We have the following relation on $\core$:
\begin{equation}
  \label{eq:20}
  \DA^2=c^2(\valpha\cdot\pa)^2+(mc^2)^2=c^2(\pa^2+{\bf \Sigma}\cdot{\bf B})+(mc^2)^2\,,
\end{equation}
where ${\bf \Sigma}$ is defined through \eqref{eq:21} and the magnetic field ${\bf B}$ is given as
\begin{equation}
  \label{eq:22}
  {\bf B}\equiv {\bf B}(\bx)=\nabla\times \V{A}(\bx)=a(\nabla\times \V{G}_\bx)+\ad (\nabla\times \V{G}_\bx)\,.
\end{equation}
Equation \eqref{eq:20} shows, in particular, that 
$
\spec(\DA) \subset (-\infty,-1]\cup[1,\infty)\,.
$

Next we state the conditions on  the potential $V$.
\begin{hypothesis}\label{hyp-V}
Let $V$ be a symmetric (matrix-valued) multiplication operator acting on $\hilbert$, such that
for almost all $\bx\in\R^3$,
$$V(\bx)=\left(\begin{array}{cc}
V_1(\bx) & 0_{2\times 2}\\
0_{2\times 2} & V_2(\bx)
\end{array}\right)\,.
$$
We assume that there exist
$a,b\in\R^+$, such that for all $\vp\in\core$,
\begin{equation}
  \label{eq:17}
  \SPb{\vp}{V^2 \,\vp}\le a^2\SPb{\vp}{\bp^2\,\vp}+ b^2\|\vp\|^2\,
\end{equation}
holds.
\end{hypothesis}
We define now the semi-relativistic Pauli-Fierz operator (minus the electron rest energy)  through  the quadratic form on $\core$,
 \begin{equation}
  \label{eq:23}
  \SPb{\vp}{\Hr{V}\vp}:= \SPb{\vp}{\big(|\DA|-\mc+V+\Hf\big)\vp}\,.
\end{equation}
Note, that on $\core$, 
\begin{equation}
  \label{eq:39}
  \Hr{0}\,\equiv\,\Hr{}=\left(\begin{array}{cc}
\hat{d}_{\V{A}}+\Hf & 0_{2\times 2}\\
0_{2\times 2} &\hat{d}_{\V{A}}+\Hf
\end{array}\right)\,,
\end{equation}
where 
\begin{equation}\label{@}
\hat{d}_{\V{A}}:=\sqrt{c^2[(\bp+\V{A})^2+{\bf \sigma}\cdot\V{B}]+(\mc)^2}-\mc
\end{equation} 
acts on $L^2(\R^3,\C^2)\otimes\mathcal{F}$.

Since the quadratic form $\SPb{\vp}{\Hr{}\vp}\equiv\SPb{\vp}{\Hr{0}\vp}$ is positive, there is a unique  (positive) self-adjoint extension of $\Hr{}\upharpoonright \core$, which we denote again by $\Hr{}$, whose
domain is contained in the  form domain of the closure of the quadratic form defined in \eqref{eq:23} for $V=0$.

We prove the following result at the end of Section \ref{estimates}.
\begin{theorem}
  \label{sa}
  Assume that $\V{G}$ fulfills Hypothesis \ref{hyp-G} and $V$ fulfills 
Hypothesis \ref{hyp-V}. Then, for $c$ (the speed of light) large enough,
the operator $ \Hr{V}\upharpoonright\core$ has a unique self-adjoint extension, denoted by 
$\Hr{V}$, with form domain $\form(\Hr{V})$ satisfying
\begin{equation}\label{fd}
\dom(\Hr{V})\subset \form(\Hr{V})=\form(\Hr{})=\form(|\bp|+\Hf)\,.
\end{equation}
Moreover, $ \Hr{V}$  is bounded below by a constant independent of $c$.
\end{theorem}
\begin{remark}
We note that a similar statement (without the specification of the form domain) was proved in \cite{MatteStockmeyer2009a} by means of a diamagnetic inequality. Our proof is based on an explicit bound for $|\bp|$ in terms of 
$\Hr{}$ (see Lemma \ref{formbound}). This bound  turns out to give the key ingredient to show \eqref{fd}.
\end{remark}
Next we define the (non-relativistic) Pauli-Fierz operator on 
$\core$ as
\begin{equation}
  \label{eq:25}
  \Hn^V:=\frac{\pa^2}{2m}+\frac{1}{2m}{\bf \Sigma}\cdot{\bf B}+V+\Hf\,,
\qquad\Hn\equiv \Hn^0\,.
\end{equation}
Actually, the operator defined  above is a two-fold
copy of the usual Pauli-Fierz operator given by $\thh:={\pa^2}/(2m)+{\bf \sigma}\cdot{\bf B}/(2m)+\Hf$ acting on $L^2(\R^3, \C^2)\otimes\mathcal{F}$:
$$\Hn^V=\left(\begin{array}{cc}
\thh+V_1 & 0_{2\times2}\\
0_{2\times 2}&\thh+V_2
\end{array}\right). $$
A potential  $V$  satisfying Hypothesis \ref{hyp-V} is  relatively  $\bp^2$-bounded
with bound zero, hence  $\Hn^V$ is essentially self-adjoint on any core 
for $\bp^2+\Hf$ and its self-adjoint extension has  domain $\dom(\Hn^V)=
\dom(\bp^2+\Hf)$ (see
\cite{HaslerHerbst2007b,Hiroshima2002}). Henceforth, we set 
\begin{equation}
  \label{eq:2}
 \TA=\frac{\pa^2}{2m}+\frac{1}{2m}{\bf \Sigma}\cdot{\bf B}\,.
\end{equation}
We  can now state the main result of this article whose proof can be found at the end of Section \ref{nrl}.
\begin{theorem}\label{main}
  Assume that ${\mathbf G}$ and $V$ satisfy Hypotheses \ref{hyp-G} and
  \ref{hyp-V} respectively. Then, for any $z\in \C\setminus\R\,$,
  \begin{equation*}
    \label{eq:6}
    \lim_{c\to\infty} \left\|(\Hr{V}-z)^{-1}- (\Hn^V-z)^{-1}\right\|=0\,.
  \end{equation*}
\end{theorem}
\noindent
The rest of the paper is organized as follows: The main technical ingredients and the proof of Theorem \ref{sa} are presented in Section \ref{estimates}. In Section \ref{nrl} we prove Theorem \ref{main} about the non-relativistic limit. Finally the 
main text is followed by an Appendix were we prove some known inequalities.
Throughout this paper we shall assume that Hypotheses \ref{hyp-G} and 
\ref{hyp-V} are fulfilled.
\section{Main estimates}\label{estimates}
The following estimates are well known. A proof of the lemma can 
be found in Appendix \ref{ap1}.
\begin{lemma}[Field-Bounds]
  \label{field-bounds}
Assume that $f, f\omega^{-1/2}\in \mathfrak{h}$. Then, for all 
$\vp \in \dom(\Hf^{1/2})$,
\begin{align}
  \label{eq:24}
  \|a(f)\vp\|^2&\le   \|f/\omega^{1/2}\|_{\mathfrak{h}}^2   \, \|\Hf^{1/2}\vp\|^2\,,\\\label{eq:24b}
     \|\ad(f)\vp\|^2&\le   \|f/\omega^{1/2}\|_{\mathfrak{h}}^2   \, \|\Hf^{1/2}\vp\|^2         +  \|f\|^2_{\mathfrak{h}}    \|\vp\|^2\,.
\end{align}
In particular, for all 
$\vp \in \dom(\Hf^{1/2})$, 
\begin{align}
  \label{eq:1}
  \SPb{\vp}{{\bf \Sigma}\cdot{\bf B}\vp}&\le 2d_1\|\Hf^{1/2}\vp\|\|\vp\|\,,\\
\label{eq:1.1}\SPb{\V{A}\vp}{ \V{A} \vp}&\le 2d_{-1}^2 \|\Hf^{1/2}\vp\|^2+d_0^2\|\vp\|^2\,.
\end{align}
\end{lemma}
For notational convenience we set
\begin{equation}
  \label{eq:5}
  \da:=|\DA|-mc^2=\{c^2(\al\cdot\pa)^2+(\mc)^2\}^{1/2}-\mc\,,
\end{equation}
acting on $L^2(\R^3,\C^4)\otimes\mathcal{F}$. (Note that $\da=\hat{d}_{\V{A}}\oplus\hat{d}_{\V{A}} $, where $\hat{d}_{\V{A}}$ is defined in \eqref{@}.)
Furthermore, 
\begin{eqnarray}
  \label{eq:40}
  &&\ra{z}:=(\Hr{V}+z)^{-1}\,,\qquad\RA:= R^{0}_{c}(z)\,,\\
&&\ran{z}:=(\Hn^V+z)^{-1}\,,\qquad\RAn:= R^{0}_{\infty}(z)\,,
\end{eqnarray}
for some $z\in\C$ with $-z\in\varrho(\Hr{V})\cap \varrho(\Hn^V)$, where $\varrho(\cdot)$ denotes the resolvent set. 
\begin{lemma}\label{lemma1}
For any $\vp\in \dom(\TAA)$, we have 
\begin{equation}
  \label{eq:4}
  \TA\vp=\big(\da+\frac{\da^2}{2mc^2}\big)\vp\,.
\end{equation}
Moreover, for $\nu \in (0,2]$, we have the relation 
$\dom(\Hn^V)=\dom(\bp^2+\Hf)\subset\dom(\da^{\,\nu})$ and the estimate, for any 
$-z\in\varrho(\Hn^V)$,
\begin{equation}
  \label{eq:13}
  \|\da^{\,\nu} \ran{z}\|= \bigO(c^\nu)\,.
\end{equation}
\end{lemma}
\begin{proof}
  Let $f: \R\to[0,\infty)$ with
  $f(t)=\sqrt{(ct)^2+(mc^2)^2}-mc^2$.  Obviously, 
  $t^2/2m=f(t)+f(t)^2/(2\mc)$.  Since $\da=f(\al\cdot\pa)$, we
  find \eqref{eq:4} by the functional calculus.

  Clearly, we have, for all $\nu\in(0,2]$, the inequality
  $f(t)^{2\nu}\le(c |t|)^{2\nu}$, which implies, for $\vp\in\core$, that
$\SPb{\vp}{\da^{2\nu}\vp}\le c^{2\nu}\SPb{\vp}{|\al\cdot\pa|^{2\nu}\vp}\,.$
In view of  Young's inequality [$ab\le a^p/p+b^q/q$, for $p,q\in(1,\infty)$ with
$1/p+1/q=1$], we obtain, for $t\in\R$ and $\nu\in(0,2)$, that $|t|^{2\nu}\le \nu t^4/2+(2-\nu)/2.$ Hence, for any $\nu\in(0,2]$, we get
\begin{equation}
  \label{eq:14}
  \SPb{\vp}{\da^{2\nu}\vp}\le  c^{2\nu} \nu/2 \SPb{\vp}{[(\al\cdot\pa)^{4}+(2-\nu)/\nu]\vp}\,.
\end{equation}
Recall that $(\al\cdot\pa)^{4}=(2m)^2 \TA^2$. Due to \eqref{eq:14} and the fact that 
\begin{equation}\label{dom}
\dom(\Hn^V)=\dom(\bp^2+\Hf)\,,
\end{equation}
we see that, in order to show \eqref{eq:13}, it suffices to prove that there exists a constant $k>0$, such that for any $\vp \in \core$,
$\|\TA\vp\|^2\le k(\|(\bp^2+\Hf)\vp\|^2+ \|\vp\|^2).$
That this is indeed the case follows from the following consideration: By \eqref{dom}, there exists a constant $k'>1$, such that $\|\Hn\vp\|^2\le k' (\|(\bp^2+\Hf)\vp\|^2+ \|\vp\|^2) $, thus, 
\begin{equation}
  \label{eq:43}
\begin{split}
  \|\TA\vp\|^2&\le 2\big( \|(\TA+\Hf)\vp\|^2+\|\Hf\vp\|^2\big)\\
 &\le 2k'\big( \|(\bp^2+\Hf)\vp\|^2+\|\Hf\vp\|^2+ \|\vp\|^2\big)\\
&\le 2\sqrt{2}k'\big( \|(\bp^2+\Hf)\vp\|^2+ \|\vp\|^2\big)\,.
\end{split}
\end{equation}
This finishes the proof.
\end{proof}
The next Lemma is proved in \cite{MatteStockmeyer2009a}. We sketch its proof in  Appendix \ref{ap1}. We define, for 
some $E>0$, the operator
\begin{equation}
  \label{eq:26}
  \wt{H}_f:=\Hf+E\,.
\end{equation}
\begin{lemma}[\cite{MatteStockmeyer2009a}]
\label{mattestock1}
 Assume that $E\ge (10\pi d_1)^2/(mc)^2$. Then, for all
$\vp\in\core$,
\begin{equation}\label{comm-posit}
\Re\SPb{|\DA|\vp}{\HT\,\vp}\,\grg\,
(1-\frac{10\pi d_1}{mc E^{1/2}})\,\big\|\,|\DA|^{1/2}\,\wt{H}^{1/2}_f\,\vp\,\big\|^2
\,.
\end{equation}
\end{lemma}
In the following Lemma we establish that the momentum $|\bp|$ is dominated, in the sense of quadratic forms, by the operator $\Hr{}\,$ (see \eqref{eq:39}).
\begin{lemma}\label{formbound}
For all $\ve\in(0,1]$ there a exist a constant $k_\ve>0$, such that 
 for all $c\ge 1$ and $\vp\in\core$, 
\begin{equation}
  \label{eq:46}
  \SPb{\vp}{|\bp|\vp}\le \sqrt{(1+\ve)(\ve+1/c^2)}\SPb{\vp}{(\da+\Hf)\vp}+k_\ve\|\vp\|^2\,.
\end{equation}
\end{lemma}
\begin{proof}
We start by observing that, as quadratic form on $\core$,
\begin{equation}
  \label{eq:10}
\begin{split}
 \bp^2&=\pa^2
-\pa\cdot\V{A}-\V{A}\cdot\pa+\V{A}^2\,.
\end{split}
\end{equation}
Estimating the cross terms and using \eqref{eq:1.1} we get, for $\ve\in(0,1]$,
\begin{equation}
  \label{eq:34}
\begin{split}
  \bp^2&\le (1+\ve) \,\pa^2+\frac{2}{\ve}
\V{A}^2\le (1+\ve) \,\pa^2+\frac{4 d_{-1}^2}{\ve}\Hf
+\frac{2d_0^2}{\ve}\,.
\end{split}  
\end{equation}
We  use \eqref{eq:1} to include the spin-term. For 
some $\delta>0$ and $\vp\in\core$, we have
\begin{equation}
  \label{eq:35}
\begin{split}
  (1+\ve)\big|\SPb{\vp}{{\bf \Sigma}\cdot\V{B}\vp}\big|&\le 4d_1 \|\Hf^{1/2}\vp\|\,\|\vp\|\\ &\le \frac{2d_1}{\delta}\SPb{\vp}{\Hf\vp}+2d_1\delta\|\vp\|^2\,\\
&\le \SPb{\vp}{\Hf\vp}+4d_1^2\|\vp\|^2\,,
\end{split}
\end{equation}
where we choose $\delta=2d_1$. Inserting \eqref{eq:35} in \eqref{eq:34} and using the identity
$\displaystyle (\al\cdot\pa)^2=\pa^2+ {\bf \Sigma}\cdot\V{B}$ we get that
there exists a constant $k>0$ such that, for all $\ve\in(0,1]$,
\begin{equation}
  \label{eq:47}
  \frac{\bp^2}{1+\ve}\le (\al\cdot\pa)^2+1+\frac{k}{\ve}(\Hf+1)\equiv \DAT^2+\frac{k}{\ve}(\Hf+1)\,,
\end{equation}
where we define $\DAT:=\al\cdot\pa+\beta=\DA|_{m=c=1}$. In view of
Lemma \ref{mattestock1} (with $m=c=1$) we have, for $\HT=\Hf+E$,
$E>\max\{1,(10\pi d_1)^2\}$, and some $\delta>0$,
\begin{equation}
  \label{eq:48}
\begin{split}
   \frac{\bp^2}{1+\ve}&\le \DAT^2+\delta(|\DAT|\HT+\HT|\DAT|)+\frac{k}{\ve}\HT\\
&\le (|\DAT|+\delta\HT)^2+\big(\frac{k}{2\ve\delta}\big)^2\,.
\end{split}
\end{equation}
Observing that from $(\al\cdot\pa)^2=2m\da+\da^2/c^2$ follows $|\al\cdot\pa|\le\sqrt{1/c^2+\ve}\,\da+m/\ve^{1/2}$ and that the square root is operator monotone, we get that
\begin{align}\nonumber
  \frac{|\bp|}{\sqrt{1+\ve}}&\le |\DAT|+\delta\HT+\frac{k}{2\ve\delta}\le
|\al\cdot\pa|+\delta\HT+\frac{k}{2\ve\delta}+1\\
&\nonumber\le \sqrt{1/c^2+\ve}\,\da+\delta\HT+\frac{k}{2\ve\delta}+1+\frac{m}{\ve^{1/2}}\\
&\label{e}\le \sqrt{1/c^2+\ve}\,(\da+\Hf)+\sqrt{2}E+\frac{k}{2\ve^{3/2}}+1+\frac{m}{\ve^{1/2}}\,,
\end{align}
where in the last inequality we choose $\delta=\sqrt{1/c^2+\ve}> \ve^{1/2}$
and we used that $\ve\le 1\le c$. Inequality \eqref{e} proves the lemma. 
\end{proof}
\begin{corollary}\label{boundV}
For any $\lambda\in(0,1]$ there exists $R_\lambda>0$, such that
for all $\rho>R_\lambda$ and  $c> \max\{2/\lambda,2a/\lambda\}$,  
\begin{equation}
  \label{eq:50}
  \||V|^{1/2}(\Hr{}+\rho)^{-1/2}\|^2\le \lambda\,,
\end{equation}
where $a$ is the constant appearing in Hypothesis \ref{hyp-V}.
\end{corollary}
\begin{proof}
We use \eqref{eq:46} and note that  \eqref{eq:17} implies  that
\begin{equation}
  \label{eq:51}
  |V|\le a|\bp|+b\le  \sqrt{2}a\sqrt{\ve+1/c^2} \Hr{} +b+ a k_\ve< \lambda
(\Hr{}+\frac{1}{\lambda}(b+ a k_\ve))\,,
\end{equation}
where we set $\ve=\min\{\lambda^2/4,\lambda^2/(4a^2)\}$ and
 we used that $1/c^2<\min\{\lambda^2/4,\lambda/(4a^2)\}$. These facts prove the claim with $R_\lambda:=(b+ a k_\ve)/\lambda$.
\end{proof}
\begin{proof}[Proof of Theorem \ref{sa}]
  For $c$ large enough we have, by \eqref{eq:51}, that $V$ is $\Hr{}$-form bounded with form bound less that one. Therefore, by the KLMN theorem (see \cite[Theorem X.17]{ReedSimonII}), the operator $\Hr{V}$ has a unique self-adjoint extension with form domain $\form(\Hr{})$ which is bounded below by $-R_{1}$, where $R_1$ is defined  after Equation \eqref{eq:51}.

In order to prove that $\form(\Hr{})=\form(|\bp|+\Hf)$, it suffices to show that
the quadratic form norms associated to $\Hr{}$ and $|\bp|+\Hf$, defined on $\core$, are equivalent. Using  the  inequality 
  \begin{equation*}
\frac{\da^2}{c^2}\le \pa^2+{\bf \Sigma}\cdot\V{B}\le 2\bp^2+2\V{A}^2+{\bf \Sigma}\cdot\V{B}\,,
  \end{equation*}
and the estimates \eqref{eq:1.1} and \eqref{eq:35}, we find some constant
$K'>0$, such that
\begin{equation*}
  \label{eq:37}
  \da^2\le (K')^2[(|\bp|+\Hf)^2+1]\,.
\end{equation*}
Therefore, by the (operator) monotonicity of the square root, we get
\begin{equation*}
  \label{eq:38}
  \da+\Hf\le K'(|\bp|+\Hf) + \Hf+K'\le (1+K')(|\bp|+\Hf+1)\,. 
\end{equation*}
To conclude the proof we note that, by Lemma \ref{formbound}, we know that there is a constant $K$, such that $|\bp|\le K(\da+\Hf+1)$. Thus, for $\delta\in(0,1)$, we have, as quadratic forms on $\core$, 
\begin{equation*}
  \label{eq:7}
  \delta(\Hf+|\bp|)\le K\da+(K+\delta)\Hf+K\le(K+\delta)(\da+\Hf+1)\,.
\end{equation*}
This finishes the proof of the assertion. 
\end{proof}
\section{The non-relativistic limit}\label{nrl}
\begin{lemma}\label{factorHf}
Assume that $c\ge1$ is so large that $\Hr{V}$ is self-adjoint. Then,  for any $z\in \C$ with $-z\in\varrho(\Hr{V})\cap \varrho(\Hn^V)$, 
 the identity  
\begin{equation}
  \label{eq:11}
\ra{z}=\big[1+\ra{z}\frac{\da^2}{2mc^2}\big]\ran{z}\,
\end{equation}
holds on $\hilbert$.
\end{lemma}
\begin{proof}
We pick some $\vp\in\core$ and 
use \eqref{eq:4} to compute
\begin{equation*}
  \label{eq:15}
  \begin{split}
\ra{z}(\Hnr{V}+z)\vp&=\ra{z}(\TA+\hf+V+z)\vp\\
&=\ra{z}\big(\da+\frac{\da^2}{2\mc}+V+\hf+z\big)\vp=\big(1+\ra{z}\frac{\da^2}{2\mc}\big)\vp\\
&=\big\{\big(1+\ra{z}\frac{\da^2}{2\mc}\big)\ran{z}\big\}(\Hnr{V}+z)\vp\,.
  \end{split}
\end{equation*}
Due to the fact that
$\Hn^V$ is essentially self-adjoint on $\core$ we find \eqref{eq:11} on some dense subset of $\hilbert$. Moreover, thanks to \eqref{eq:13},  the term in $\{\cdots\}$ above is bounded. Therefore, by a simple limit argument, we obtain \eqref{eq:11}.
 \end{proof}
\begin{proof}[Proof of Theorem \ref{main}]
We first note that both operators, $\Hn^V$ and $\Hr{V}$, are bounded below by a constant independent of the speed of light $c$. In fact, by \eqref{eq:4}, we have that  $\Hn^V\ge\Hr{V}$ the last operator being bound\-ed below (independent of $c$) due to Theorem \ref{sa}.
We define $e_V:=\inf(\sigma(\Hr{V}))$ and choose $\rho>|e_V|$.

Next, we observe that it is enough to prove Theorem \ref{main} for $\rho$ instead of $z$. Indeed,  define $\mathcal{S}_\varkappa:=(H_\varkappa^V+\rho)/(H_\varkappa^V+z)$, for $z\in \C\setminus \R$ and $\varkappa=c$ or $\infty$. Clearly,
  $\|\mathcal{S}_\varkappa\|\le 1+(\rho+|z|)/|\Im{z}|$.
Further, for $z\in \C\setminus \R$, we have the following  well-known operator identity
\begin{equation}\label{**}
  \ra{z}-\ran{z}=\mathcal{S}_c \,(\ra{\rho}-\ran{\rho})\, \mathcal{S}_\infty\,.
\end{equation}
Therefore, on account of Lemma \ref{factorHf} and \eqref{**}, it suffices to show that
\begin{equation}\label{lim}
\|\ra{\rho}\frac{\da^2}{2mc^2}\ran{\rho}\|\to 0\qquad\mbox{ as}\qquad c\to\infty\,.
\end{equation}
We observe 
that by Tiktopolus' Formula (see \cite{Simon1971Book}) and \eqref{eq:50}, we have, for  $c\ge \max\{4,4a\}$, that
\begin{equation}
  \label{eq:12}
  \ra{\rho}=\raf^{1/2}\big(1+\raf^{1/2}V\raf^{1/2}\big)^{-1}\raf^{1/2}\,.
\end{equation}
Therefore, using the obvious fact that $\|\da^{1/2}\raf^{1/2}\|\le 1$, we find
\begin{equation}
  \label{eq:18}
\begin{split}
  \|\ra{\rho}\frac{\da^2}{2mc^2} \ran{\rho} \|&\le \frac{\|\raf^{1/2}\|\,\|\da^{1/2}\raf^{1/2}\|\,\|\da^{3/2}\ran{\rho}\|}
{2\mc (1-\||V|^{1/2}\raf^{1/2}\|^2)}\\
&\le \frac{\|\da^{3/2}\ran{\rho}\|}{\mc\sqrt{\rho-|e_V|}}=\bigO(c^{-1/2})\,,
\end{split}
\end{equation}
where in the last estimate we use \eqref{eq:13}. This proves \eqref{lim}
and concludes the proof of the theorem. 
\end{proof}

\bigskip

\noindent
{\bf Acknowledgement:} I am grateful to Oliver Matte and Thomas
 \O stergaard S\o rensen for  proofreading the manuscript and helpful discussions.
This work has been partially supported
by the DFG (SFB/TR12).
\appendix
\section{}\label{ap1}
\begin{proof}[Proof of Lemma \ref{field-bounds}]
Let $\vp\in\sC_0$ and $f, f\omega^{-1/2}\in \mathfrak{h}$. Using \eqref{*} and the Cauchy-Schwarz inequality, we get that
\begin{equation}
  \label{eq:3}
\begin{split}
  \|a(f)\vp\|&=\big\|\int \overline{f}(k) a(k) \vp \,dk\big\|\le \int
\frac{|f(k)|}{\omega(k)^{1/2}}\,\|\omega(k)^{1/2}a(k)\vp\|\,dk\\
&\le \|f\omega^{-1/2}\|_{\HP}\,\Big(\int \|\omega(k)^{1/2}a(k)\vp\|^2\,dk\Big)^{1/2}\,,
\end{split}
\end{equation}
from which \eqref{eq:24} follows,  after using that
\begin{equation}\label{Hf=dGamma}
  \SPb{H_f^{1/2} \phi}{H_f^{1/2} \psi} = \int\omega(k) 
  \SPn{a(k) \phi}{a(k) \psi} dk\,,
  \qquad \phi,\psi\in\dom(\Hf^{1/2})\,,
\end{equation}
which is a consequence of Fubini's theorem. Equation 
\eqref{eq:24b} follows from the fact that
\begin{equation}
  \label{eq:9}
\|\ad(f)\vp\|^2=\SPb{\vp}{a(f)\ad(f)\vp}= \|a(f)\vp\|^2+\|f\|^2_{\mathfrak{h}}    \|\vp\|^2\,, 
\end{equation}
where we used the commutation relations \eqref{ee}.

Finally equations \eqref{eq:1} and \eqref{eq:1.1} follow from \eqref{eq:24} and \eqref{eq:24b}, the definitions 
\eqref{dd} and \eqref{eq:22}, and the Hypothesis \ref{hyp-G}.
\end{proof}
We shortly present the proof Lemma  \ref{mattestock1}
which can be found in \cite[Lemma 4.1]{MatteStockmeyer2009a}. We will keep the notation of \cite{MatteStockmeyer2009a}.
\begin{proof}[Proof of Lemma \ref{mattestock1}]
To begin with we observe that by means of the pull-through formula we have
(see \cite[Lemma 3.1]{MatteStockmeyer2009a}), 
for $\HT=\Hf+E$ and $E>0$,
\begin{equation}
\|T_{1/2}\|:=\big\| [\valpha\cdot\V{A}, \HT^{-1/2}] \HT^{1/2} \big\|
 \klg 
2 d_1/E^{1/2}
\,.\label{combound2}
\end{equation}
We note that a simple computation shows (see \cite[Corollary 3.2]{MatteStockmeyer2009a}), for $R_\eta:=1/(\DA-i\eta)$, $\eta\in\R$, and 
$\psi\in\hilbert$, that
\begin{equation}
  \label{eq:27}
  R_\eta\HT^{-1/2}\psi=(1+cR_\eta T_{1/2})\HT^{-1/2}R_\eta\psi\,,
\end{equation}
which implies, using that $\|R_\eta\|=1/\sqrt{((\mc)^2+\eta^2)}$, that
\begin{equation}
  \label{eq:28}
  \|\HT^{-1/2}R_\eta \HT^{1/2}\|\le\frac{1}{\sqrt{(\mc)^2+\eta^2}}\frac{1}{1-2d_1/(mcE^{1/2})}\,,
\end{equation}
provided $E>(2d_1/(mc))^2$.

Now, we use \eqref{eq:28} to obtain a bound for $
  S_{1/2}:=[\sgn(\DA), \HT^{-1/2}] \HT^{1/2}\,$.
Applying the formula (see \cite[Lemma~VI.5.6]{Kato})
$
\sgn(\DA)\varphi= \lim_{\tau\to\infty} 
\int_{-\tau}^\tau R_\eta \varphi\, \frac{d\eta}{\pi}$, for $
\psi\in\HR
$, we find, for $\vp\in\core$,
\begin{equation}
  \label{eq:30}
\begin{split}
  S_{1/2}\vp&=c\int_{-\infty}^\infty [R_\eta, \HT^{-1/2}] \HT^{1/2}\vp \,\frac{d\eta}{\pi}=c\int_{-\infty}^\infty R_\eta T_{1/2} \HT^{-1/2}R_\eta \HT^{1/2}
\vp \,\frac{d\eta}{\pi}\,.
\end{split}
\end{equation}
From this follows, using \eqref{eq:28}, \eqref{combound2},
and $\int \|R_\eta\|^2 d\eta/\pi=1/\mc$, that
\begin{equation}
  \label{eq:31}
  \|S_{1/2}\|\le \frac{2d_1}{mcE^{1/2}}\frac{1}{1-2d_1/(mcE^{1/2})}\,,
\quad\mbox{for}\quad E>(2d_1/(mc))^2\,.
\end{equation}
 An analogous calculation and the fact that$
  \int  \| |\DA|^{1/2}R_\eta \|\, \| R_\eta\|    d\eta/\pi\le 2\pi/\sqrt{\mc},$
show that
\begin{equation}
  \label{eq:33}
   \||\DA|^{1/2} S_{1/2}\|\le \frac{4\pi d_1}{\sqrt{m}E^{1/2}}\frac{1}{1-2d_1/(mcE^{1/2})}\,,
\quad\mbox{for}\quad E>(2d_1/(mc))^2\,.
\end{equation}
Finally, we argue as in \cite[Lemma 4.1]{MatteStockmeyer2009a}: 
Let $\phi\in\core$ and set $\psi:=\HT^{1/2}\,\phi$.
We have
\begin{align*}
\lefteqn{
\Re \SPb{\DA \HT^{-1/2} \psi}{\sgn(\DA) \HT^{1/2} \psi}
}
\\
&=
\Re \SPb{(\DA-c\,T^*_{1/2}) \psi}{\HT^{-1/2} \sgn(\DA) \HT^{1/2} \psi}
\\
&=
\Re \SPb{(\DA-c\,T^*_{1/2}) \psi}{\big(\sgn(\DA)-S_{1/2}\big) \psi}
\\
&\ge
\SPn{|\DA| \psi}{\psi}- 
\big\| |\DA|^{1/2} S_{1/2}|\DA|^{-1/2} \big\|\,\| |\DA|^{1/2} \psi \|^2 
\\
& \qquad\qquad\qquad
- c\|T_{1/2}\|\,\||\DA|^{-1/2}\|^2 (1+\|S_{1/2}\|) \||\DA|^{1/2}\psi\|^2.
\end{align*}
A rough estimate shows, choosing $E>4(2d_1)^2/(mc)^2$ and using \eqref{combound2}, \eqref{eq:31}, and \eqref{eq:33}, that
\begin{equation}
  \label{eq:8}
  \Re\SPb{|\DA|\phi}{\HT\,\phi}\,\grg\,
(1-\frac{10\pi d_1}{mc E^{1/2}})\,\big\|\,|\DA|^{1/2}\,\wt{H}^{1/2}_f\,\phi\,\big\|^2
\,.
\end{equation}
This finishes the proof of Lemma \ref{mattestock1}.
\end{proof}

%
\end{document}